\documentclass[12pt,english]{article}
\usepackage[T1]{fontenc}
\usepackage[latin9]{inputenc}
\usepackage[letterpaper]{geometry}
\usepackage{textcomp}
\usepackage{amsthm}
\usepackage{amsmath}
\usepackage{amssymb}
\usepackage{graphicx}
\usepackage{setspace}
\usepackage[authoryear]{natbib}
\makeatletter
\newcommand{\lyxaddress}[1]{
\par {\raggedright #1
\vspace{1.4em}
\noindent\par}
}
 \theoremstyle{definition}
  \newtheorem{example}{\protect\examplename}
  \theoremstyle{plain}
  \newtheorem{lem}{\protect\lemmaname}
\theoremstyle{plain}
\newtheorem{thm}{\protect\theoremname}
  \theoremstyle{plain}
  \newtheorem{cor}{\protect\corollaryname}
  \theoremstyle{remark}
  \newtheorem{rem}{\protect\remarkname}

\DeclareMathOperator{\convex}{co}
\DeclareMathOperator{\extreme}{ext}

\DeclareMathOperator{\pri}{prior}

\DeclareMathOperator{\exaend}{\blacktriangle}

\makeatother

\usepackage{babel}
  \providecommand{\examplename}{Example}
  \providecommand{\lemmaname}{Lemma}
  \providecommand{\remarkname}{Remark}
\providecommand{\corollaryname}{Corollary}
\providecommand{\theoremname}{Theorem}

\begin{document}

\title{Resolving conflicts between statistical methods by probability combination: Application to empirical Bayes analyses of genomic data}

\maketitle
~\\
David R. Bickel

\lyxaddress{Ottawa Institute of Systems Biology\\
Department of Biochemistry, Microbiology, and Immunology\\
Department of Mathematics and Statistics\\
University of Ottawa; 451 Smyth Road; Ottawa, Ontario, K1H 8M5}
\begin{abstract}
In the typical analysis of a data set, a single method is selected
for statistical reporting even when equally applicable methods yield
very different results. Examples of equally applicable methods can
correspond to those of different ancillary statistics in frequentist
inference and of different prior distributions in Bayesian inference.
More broadly, choices are made between parametric and nonparametric
methods and between frequentist and Bayesian methods. Rather than
choosing a single method, it can be safer, in a game-theoretic sense,
to combine those that are equally appropriate in light of the available
information. Since methods of combining subjectively assessed probability
distributions are not objective enough for that purpose, this paper
introduces a method of distribution combination that does not require
any assignment of distribution weights. It does so by formalizing
a hedging strategy in terms of a game between three players: nature,
a statistician combining distributions, and a statistician refusing
to combine distributions. The optimal move of the first statistician
reduces to the solution of a simpler problem of selecting an estimating
distribution that minimizes the Kullback-Leibler loss maximized over
the plausible distributions to be combined. The resulting combined
distribution is a linear combination of the most extreme of the distributions
to be combined that are scientifically plausible. The optimal weights are close enough to each other that no extreme
distribution dominates the others. The new methodology is illustrated by combining conflicting empirical
Bayes methodologies in the context of gene expression data analysis.

\end{abstract}
\newpage{}

\section{\label{sec:Introduction}Introduction}

The analysis of biological data often requires choices between methods
that seem equally applicable and yet that can yield very different
results. This occurs not only with the notorious problems in frequentist
statistics of conditioning on one of multiple ancillary statistics
and in Bayesian statistics of selecting one of many appropriate priors,
but also in choices between frequentist and Bayesian methods, in whether
to use a potentially powerful parametric test to analyze a small sample
of unknown distribution, in whether and how to adjust for multiple
testing, and in whether to use a frequentist model averaging procedure.
Today, statisticians simultaneously testing thousands of hypotheses
must often decide whether to apply a multiple comparisons procedure
using the assumption that the p-value is uniform under the null hypothesis
(theoretical null distribution) or a null distribution estimated from
the data (empirical null distribution). While the empirical null reduces
estimation bias in many situations \citep{RefWorks:208}, it also
increases variance \citep{efron_large-scale_2010} and substantially
increases bias when the data distributions have heavy tails \citep{conditional2009}.
Without any strong indication of which method can be expected to perform
better for a particular data set, combining their estimated false
discovery rates or adjusted p-values may be the safest approach.

Emphasizing the reference class problem, \citet{DiversityEvidence1995}
pointed out the need for ways to assess the evidence in the diversity
of statistical inferences that can be drawn from the same data. Previous
applications of p-value combination have included combining inferences
from different ancillary statistics \citep{ISI:A1984TE72700008},
combining inferences from more robust procedures with those from procedures
with stronger assumptions, and combining inferences from different
alternative distributions \citep{RefWorks:1016}. However, those combination
procedures are only justified by a heuristic Bayesian argument and
have not been widely adopted. To offer a viable alternative, the problem
of combining conflicting methods is framed herein in terms of probability
combination. 

Most existing methods of automatically combining probability distributions
have been designed for the integration of expert opinions. For example,
\citet{Toda1956}, \citet{Abbas2009}, and \citet{Kracik2011} proposed
combining distributions to minimize a weighted sum of Kullback-Leibler
divergences from the distributions being combined, with the weights
determined subjectively, e.g., by the elicitation of the opinions
of the experts who provided the distributions or by the extent to
which each expert is considered credible. Under broad conditions,
that approach leads to the linear combination of the distributions
that is defined by those weights \citep{Toda1956,Kracik2011}. 

{}``Linear opinion pools'' also result from this \emph{marginalization
property}: any linearly combined marginal distribution is the same
whether marginalization or combination is carried out first \citep{Mcconway1981}.
The marginalization property forbids certain counterintuitive combinations
of distributions, including any combination of distributions that
differs in a probability assignment from the unanimous assignment
of all distributions combined \citep[p. 173]{cooke_experts_1995}.
Combinations violating the marginalization property can be expected
to perform poorly as estimators regardless of their appeal as distributions
of belief. On the other hand, invariance to reversing the order of
Bayesian updating and distribution combination instead requires a
{}``logarithmic opinion pool,'' which uses a geometric mean in place
the arithmetic mean of the linear opinion pool; see, e.g., \citet[§4.11.1]{RefWorks:179}
or \citet{RefWorks:521}. While that property is preferable to the
marginalization property from the point of view of a Bayesian agent
making decisions on the basis of independent reports of other Bayesian
agents, it is less suitable for combining distributions that are highly
dependent or that are distribution estimates rather than actual distributions
of belief. \citet{Genest1986} and \citet[Ch. 11]{cooke_experts_1995}
review these and related issues. 

Like those methods, the strategy introduced in this paper is intended
for combining distributions based on the same data or information
as opposed to combining distributions based on independent data sets.
However, to relax the requirement that the distributions be provided
by experts, the weights are optimized rather than specified. While
the new strategy leads to a linear combination of distributions, the
combination hedges by including only the most extreme distributions
rather than all of the distributions. In addition, the game leading
to the hedging takes into account any known constraints on the true
distribution. See Remark \ref{rem:game-theory} on the pivotal role
game theory played in the foundations of statistics.

The game that generates the hedging strategy is played between three
players: the mechanism that generates the true distribution ({}``Nature''),
a statistician who never combines distributions ({}``Chooser''),
and a statistician open to combining distributions ({}``Combiner'').
Nature must select a distribution that complies with constraints known
to the statisticians, who want to choose distributions as close as
possible to the distribution chosen by Nature. Other things being
equal, each statistician would also like to select a distribution
that is as much better than that of the other statistician as possible.
Thus, each statistician seeks primarily to come close to the truth
and secondarily to improve upon the distribution selected by the other
statistician. Combiner has the advantage over Chooser that the former
may select any distribution, whereas the latter must select one from
a given set of the distributions that estimate the true distribution
or that encode expert opinion. On the other hand, Combiner is disadvantaged
in that Nature seeks to maximize the gain of Chooser without concern
for the gain of Combiner. Since Nature favors Chooser without opposing
Combiner, the optimal strategy of Combiner is one of hedging but is
less cautious than the minimax strategies that are often optimal for
typical two-player zero-sum games against Nature. The distribution
chosen according to the strategy of Combiner will be considered the
combination of the distributions available to Chooser. The combination
distribution is a function not only of the combining distributions
but also of the constraints on the true distribution.

Sec. \ref{sec:Framework-for-combining} encodes the game and strategy
described above in terms of Kullback-Leibler loss and presents its
optimal solution as a general method of combining distributions. The
important special case of combining probabilities is then worked out.
A framework for using the proposed combination method to resolve
method conflicts in point and interval estimation, hypothesis testing,
and other aspects of statistical data analysis will be presented in
Sec. \ref{sec:Statistical-inference}. The framework is illustrated
by applying it to the combination of three false discovery rate methods
for the analysis of microarray data in Sec. \ref{sec:Large-scale-case-study}.
Finally, Appendices A and B collect miscellaneous remarks and proofs,
respectively.

\section{\label{sec:Framework-for-combining}Framework for combining distributions}

\subsection{\label{sub:Information-theory}Information-theoretic background}

Let $\mathcal{P}$ denote the set of probability distributions on
a Borel space $\left(\Xi,\mathcal{B}\left(\Xi\right)\right)$, where
$\mathcal{B}\left(\Xi\right)$ is the set of all Borel subsets of
$\Xi$. The \emph{information divergence} \emph{of $P\in\mathcal{P}$
with respect to $Q\in\mathcal{P}$} is defined as 
\begin{equation}
D\left(P||Q\right)=\int dP\left(\xi\right)\log\frac{dP\left(\xi\right)}{dQ\left(\xi\right)},\label{eq:cross-entropy}
\end{equation}
where $dP$ and $dQ$ are probability density functions of $P$ and
$Q$ in the sense of Radon-Nikodym differentiation with respect to
the same dominating measure \citep{Haussler1997b}. The integrand
follows the $0\log\left(0\right)=0$ and $0\log\left(0/0\right)=0$
conventions. $D\left(P||Q\right)$ is also known as {}``information
for discrimination,'' {}``Kullback-Leibler information,'' {}``Kullback-Leibler
divergence,'' and {}``cross entropy.'' Calling $D\left(P||Q\right)$
{}``information divergence'' emphasizes its interpretation as the
amount of information that would be gained by replacing any distribution
$Q$ with the true distribution $P$. That interpretation accords
with calling 
\begin{equation}
D\left(P^{\prime}||P^{\prime\prime}\rightsquigarrow Q\right)=D\left(P^{\prime}||P^{\prime\prime}\right)-D\left(P^{\prime}||Q\right)\label{eq:information-gain}
\end{equation}
the \emph{information gain} \citep{Phuber77}, the amount of information
gained by using $Q$ rather than $P^{\prime\prime}\in\mathcal{P}$
when the true distribution is $P^{\prime}\in\mathcal{P}$ \citep[cf.][]{Topsoe2007b}. 

For any real parameter set $\Phi$ and family $\mathcal{P}^{\star}=\left\{ P_{\phi}:\phi\in\Phi\right\} $
of probability distributions such that $\mathcal{P}^{\star}\subseteq\mathcal{P}$,
the distribution 
\begin{equation}
\widetilde{\mathcal{P}^{\star}}=\arg\inf_{Q\in\mathcal{P}}\sup_{P^{\star}\in\mathcal{\mathcal{P}^{\star}}}D\left(P^{\star}||Q\right)\label{eq:centroid}
\end{equation}
is called the \emph{centroid} of $\mathcal{P}^{\star}$ \citep[p. 131]{Csiszar2011}.
Let $\mathcal{W}$ denote the set of all measures on the Borel space
$\left(\Phi,\mathcal{B}\left(\Phi\right)\right)$. Reserving the term
\emph{prior} for Sec. \ref{sec:Statistical-inference}, members of
$\mathcal{W}$ will be called \emph{weighting distributions}. Then
$P^{W}=E_{W}\mathcal{P}^{\star}=\int P_{\phi}dW\left(\phi\right)$
defines the \emph{mixture distribution} \emph{of }$\mathcal{P}^{\star}$
\emph{with respect to} \emph{some} $W\in\mathcal{W}$, and
\begin{equation}
W_{\mathcal{P}^{\star}}=\arg\sup_{W\in\mathcal{W}}\int D\left(P_{\phi}||E_{W}\mathcal{P}^{\star}\right)dW\left(\phi\right)\label{eq:weighting-distribution}
\end{equation}
defines the \emph{weighting distribution induced by }$\mathcal{P}^{\star}$. 
\begin{example}
In the case of a family of $\nu$ distributions, the parameter set
can be written as $\Phi=\left\{ \phi_{1},\dots,\phi_{\nu}\right\} $
and the weighting distribution as 
\[
\left\langle W_{\mathcal{P}^{\star}}\left(\phi_{1}\right),\dots,W_{\mathcal{P}^{\star}}\left(\phi_{\nu}\right)\right\rangle =\arg\sup_{w_{1}\in\left[0,1\right],\dots,w_{\nu}\in\left[0,1\right]}\sum_{i=1,\dots,\nu}w_{i}D\left(P_{\phi}||\sum_{j=1,\dots,\nu}w_{j}P_{\phi_{j}}\right),
\]
where the supremum is that of the set of weight $\nu$-tuples constrained
by $\sum_{i=1}^{\nu}w_{i}=1$. \citet{Shulman2004} proved that,
for all $i=1,\dots,\nu$, 
\begin{equation}
W_{\mathcal{P}^{\star}}\left(\phi_{i}\right)\le1-e^{-1}\doteq63\%.\label{eq:weight-bound}
\end{equation}
$\exaend$
\end{example}
The next known result will prove useful in determining the optimal
move in the game of combining distributions that was mentioned in
Sec. \ref{sec:Introduction}.
\begin{lem}
\label{lem:minimax-redundancy}The centroid of any nonempty $\mathcal{P}^{\star}\subseteq\mathcal{P}$
is $\widetilde{\mathcal{P}^{\star}}=P^{W_{\mathcal{P}^{\star}}}$,
where $W_{\mathcal{P}^{\star}}$ is the weighting distribution induced
by $\mathcal{P}^{\star}$.\end{lem}
\begin{proof}
Two different proofs appear in \citet{Haussler1997b} and in \citet{Gruenwald20041367}.
For some history of this result, see Remark \ref{rem:minimax-redundancy}.
\end{proof}

\subsection{\label{sub:Distribution-combination-game}Distribution-combination
game}

The game sketched in Sec. \ref{sec:Introduction} will now be specified
in the above notation. Two sets constrain moves in the game: the \emph{plausible
set} $\dot{\mathcal{P}}$ is the subset of $\mathcal{P}$ consisting
of given \emph{plausible distributions}, and $\ddot{\mathcal{P}}\subseteq\mathcal{P}$
consists of given \emph{combining distributions}. The move of Nature
is a distribution $\dot{P}\in\dot{\mathcal{P}}$; the move of Chooser
is a distribution $\ddot{P}\in\ddot{\mathcal{P}}$; the move of Combiner
is a distribution $P^{+}\in\mathcal{P}$. Chooser and Combiner are
called \emph{statisticians}. If $P_{1}$ is the move of one statistician
and $P_{2}$ is that of the other, then the amount of utility paid
to the latter is the pair
\begin{equation}
U\left(\dot{P},P_{1},P_{2}\right)=\left\langle -D\left(\dot{P}||P_{2}\right),D\left(\dot{P}||P_{1}\rightsquigarrow P_{2}\right)\right\rangle ,\label{eq:utility-function}
\end{equation}
understood in terms of preferring $v=\left\langle v_{1},v_{2}\right\rangle \in\left[0,\infty\right)\times\left[0,\infty\right)$
over $u=\left\langle u_{1},u_{2}\right\rangle \in\left[0,\infty\right)\times\left[0,\infty\right)$
if and only if $u\preceq v$. Here, $u\preceq v$ means that either
$u_{1}<v_{1}$ or both $u_{1}=v_{1}$ and $u_{2}\le v_{2}$. Such
preferences are said to have \emph{lexicographic ordering} (Remark
\ref{rem:lexicographic}).

Thus, the utility paid to Combiner will be $U\left(\dot{P},\ddot{P},P^{+}\right)$
and that paid to Chooser will be $U\left(\dot{P},P^{+},\ddot{P}\right)$.
The utility paid to Nature will also be $U\left(\dot{P},P^{+},\ddot{P}\right)$,
with the implication that it is to the advantage of Nature and Chooser
to act as a \emph{coalition} with move $\left\langle \dot{P},\ddot{P}\right\rangle $
\citep[Ch. 5]{MaxUtility1944}. Although that reduces the three-player
game to a two-player game of the coalition versus Combiner, it is
not necessarily of zero sum.

The \emph{combination of the distributions in} $\ddot{\mathcal{P}}$
\emph{with truth constrained by} $\dot{\mathcal{P}}$ is defined as
Combiner's optimal move in the game. Since the utility paid to the
Nature-Chooser coalition is $U\left(\dot{P},P^{+},\ddot{P}\right)$,
Combiner's best move may be written as 
\begin{equation}
P^{+}=\arg\sup_{Q\in\mathcal{P}:\left\langle \dot{P}_{Q},\ddot{P}_{Q}\right\rangle \in\mathcal{P}_{Q}}^{\preceq}U\left(\dot{P}_{Q},\ddot{P}_{Q},Q\right)\label{eq:combined-distribution}
\end{equation}
for all $Q\in\mathcal{P}$, where  $\sup^{\preceq}$ is the least
upper bound according to $\preceq$, and
\begin{equation}
\mathcal{P}_{Q}=\arg\sup_{\left\langle P^{\prime},P^{\prime\prime}\right\rangle \in\dot{\mathcal{P}}\times\ddot{\mathcal{P}}}^{\preceq}U\left(P^{\prime},Q,P^{\prime\prime}\right).\label{eq:opponent-moves}
\end{equation}
While $P^{+}$ is not necessarily a plausible distribution, it is
typically at the center of the plausible set:
\begin{thm}
\label{thm:combination}Let $P^{+}$ denote the combination of the
distributions in $\ddot{\mathcal{P}}$ with truth constrained by $\dot{\mathcal{P}}$.
If $\dot{\mathcal{P}}\cap\ddot{\mathcal{P}}\ne\emptyset$, then 
\begin{equation}
P^{+}=\widetilde{\dot{\mathcal{P}}\cap\ddot{\mathcal{P}}}=P^{W_{\dot{\mathcal{P}}\cap\ddot{\mathcal{P}}}},\label{eq:combination}
\end{equation}
where $\widetilde{\dot{\mathcal{P}}\cap\ddot{\mathcal{P}}}$ is the
centroid of $\dot{\mathcal{P}}\cap\ddot{\mathcal{P}}$, and $W_{\dot{\mathcal{P}}\cap\ddot{\mathcal{P}}}$
is the weighting distribution induced by $\dot{\mathcal{P}}\cap\ddot{\mathcal{P}}$,
as defined by eq. \eqref{eq:weighting-distribution}.
\end{thm}
Let $\mathcal{A}$ denote an action space. A decision made by taking
the action $a\in\mathcal{A}$ that minimizes the expectation value
of a loss function $L:\Xi\times\mathcal{A}\rightarrow\mathbb{R}$
with respect to $P^{+}$ is optimal in the game when the utility
function of eq. \eqref{eq:utility-function} is replaced with 
\[
\left\langle -D\left(\dot{P}||P_{2}\right),D\left(\dot{P}||P_{1}\rightsquigarrow P_{2}\right),-\int L\left(\xi,a\right)dP_{2}\left(\xi\right)\right\rangle .
\]
The latter utility function is understood in terms of the lexicographic
ordering relation $\preccurlyeq$, which is defined such that $\left\langle u_{1},u_{2},u_{3}\right\rangle \preccurlyeq\left\langle v_{1},v_{2},v_{3}\right\rangle $
if and only if one of the following is true: (i) $u_{1}<v_{1}$; (ii)
$u_{1}=v_{1}$ and $u_{2}<v_{2}$; (iii) $u_{1}=v_{1}$, $u_{2}=v_{2}$,
and $u_{3}\le v_{3}$. On related orderings in the literature, see
Remark \ref{rem:lexicographic}.

\subsection{Combining discrete distributions}

Now let $\mathcal{P}$ denote the set of probability distributions
on $\left(\Xi,2^{\Xi}\right),$ where $\Xi$ is a finite set written
as $\Xi=\left\{ 0,1,...,\left|\Xi\right|-1\right\} $ without loss
of generality. Then the information divergence of $P$ with respect
to $Q$ \eqref{eq:cross-entropy} reduces to
\[
D\left(P||Q\right)=\sum_{i\in\Xi}P\left(\left\{ i\right\} \right)\log\frac{P\left(\left\{ i\right\} \right)}{Q\left(\left\{ i\right\} \right)}.
\]
For any $P\in\mathcal{P}$ and $\xi\sim P$, the $\left|\Xi\right|$-tuple
$T\left(P\right)=\left\langle P\left(\xi=0\right),P\left(\xi=1\right),\dots,P\left(\xi=\left|\Xi\right|-1\right)\right\rangle $
will be called the \emph{tuple representing }$P$. 

Let $\mathcal{P}^{\star}$ denote a nonempty subset of $\mathcal{P}$,
and let $\mathcal{T}\left(\mathcal{P}^{\star}\right)$ denote the
set of tuples representing the members of $\mathcal{P}^{\star}$,
i.e., $\mathcal{T}\left(\mathcal{P}^{\star}\right)=\left\{ T\left(P^{\star}\right):P^{\star}\in\mathcal{P}^{\star}\right\} $.
Likewise, noting that the map $\mathcal{T}$ is one-to-one, the \emph{extreme
subset} \emph{of }$\mathcal{P}^{\star}$ is defined as $\extreme\mathcal{P}^{\star}=\mathcal{T}^{-1}\left(\extreme\convex\mathcal{T}\left(\mathcal{P}^{\star}\right)\right),$
where $\extreme\convex\mathcal{T}\left(\mathcal{P}^{\star}\right)$
is the set of extreme points of $\convex\mathcal{T}\left(\mathcal{P}^{\star}\right)$,
the convex hull of $\mathcal{T}\left(\mathcal{P}^{\star}\right)$.
The extreme subset simplifies the problem of locating a centroid:
\begin{lem}
\label{lem:extreme-set}Let \textup{$\mathcal{P}^{\star}$ denote
a nonempty, finite subset of} $\mathcal{P}$. If there are a $Q\in\mathcal{P}$
and a $C>0$ such that $D\left(P^{\star}||Q\right)=C$ for all $P^{\star}\in\extreme\mathcal{P}^{\star}$,
then $Q$ is the centroid of $\mathcal{P}^{\star}$. 
\end{lem}
More simplification is possible if at least one of the combining distributions
is plausible:
\begin{thm}
\label{thm:extreme}Let $P^{+}$ denote the combination of the distributions
in $\ddot{\mathcal{P}}$ with truth constrained by $\dot{\mathcal{P}}$.
If $\dot{\mathcal{P}}\cap\ddot{\mathcal{P}}$ is nonempty and finite,
then $P^{+}=P^{W_{\extreme\left(\dot{\mathcal{P}}\cap\ddot{\mathcal{P}}\right)}},$
where $W_{\extreme\left(\dot{\mathcal{P}}\cap\ddot{\mathcal{P}}\right)}$
is the weighting distribution induced \eqref{eq:weighting-distribution}
by $\extreme\left(\dot{\mathcal{P}}\cap\ddot{\mathcal{P}}\right)$,
the extreme subset of $\dot{\mathcal{P}}\cap\ddot{\mathcal{P}}$. 
\end{thm}
The combination of a set of probabilities of the same hypothesis (§\ref{sec:Statistical-inference})
or event is simply a linear combination or mixture of the highest
and lowest of the plausible probabilities in the set, with the mixing
proportion determined optimally:
\begin{cor}
\label{cor:combining-probabilities}Let $P^{+}$ denote the combination
of the distributions in $\ddot{\mathcal{P}}$ with truth constrained
by $\dot{\mathcal{P}}$. Suppose $c$ distributions on $\left(\left\{ 0,1\right\} ,2^{\left\{ 0,1\right\} }\right)$
are to be combined $\left(\ddot{\mathcal{P}}=\left\{ \ddot{P}_{1},...,\ddot{P}_{c}\right\} \right)$,
and let $\dot{\mathcal{P}}_{0}=\left\{ \dot{P}\left(\left\{ 0\right\} \right):\dot{P}\in\dot{\mathcal{P}}\right\} $
and $\underline{\ddot{P}},\overline{\ddot{P}}\in\mathcal{P}$ such
that
\[
\underline{\ddot{P}}\left(\left\{ 0\right\} \right)=\min_{i=1,...,c:\ddot{P}_{i}\left(\left\{ 0\right\} \right)\in\dot{\mathcal{P}}_{0}}\ddot{P}_{i}\left(\left\{ 0\right\} \right);\,\overline{\ddot{P}}\left(\left\{ 0\right\} \right)=\max_{i=1,...,c:\ddot{P}_{i}\left(\left\{ 0\right\} \right)\in\dot{\mathcal{P}}_{0}}\ddot{P}_{i}\left(\left\{ 0\right\} \right).
\]
If $\ddot{P}_{i}\left(\left\{ 0\right\} \right)\in\dot{\mathcal{P}}_{0}$
for some $i\in\left\{ 1,...,c\right\} $, then $P^{+}=w^{+}\underline{\ddot{P}}+\left(1-w^{+}\right)\overline{\ddot{P}}$,
where 
\begin{equation}
w^{+}=\arg\sup_{w\in\left[0,1\right]}\left(w\Delta\left(\underline{\ddot{P}}||w\right)+\left(1-w\right)\Delta\left(\overline{\ddot{P}}||w\right)\right);\label{eq:probability-weight}
\end{equation}
\[
\Delta\left(\bullet||w\right)=D\left(\bullet||w\underline{\ddot{P}}+\left(1-w\right)\overline{\ddot{P}}\right).
\]
\end{cor}
\begin{proof}
This follows immediately from Theorem \ref{thm:extreme} and the definition
of an extreme subset.
\end{proof}
By eq. \eqref{eq:weight-bound}, $37\%\dot{\le}w^{+}\dot{\le}63\%$,
implying that $P^{+}\left(\left\{ 0\right\} \right)$ is close to
the arithmetic mean $\left[\underline{\ddot{P}}\left(\left\{ 0\right\} \right)+\overline{\ddot{P}}\left(\left\{ 0\right\} \right)\right]/2$,
as \citet{Shulman2004} observed in a coding context. Fig. \ref{fig:weight}
plots $w^{+}$ versus $\underline{\ddot{P}}\left(\left\{ 0\right\} \right)$
and $\overline{\ddot{P}}\left(\left\{ 0\right\} \right)$, and Fig.
\ref{fig:means} compares the resulting $P^{+}\left(\left\{ 0\right\} \right)$
to the arithmetic mean, the geometric mean, and the harmonic mean
of $\underline{\ddot{P}}\left(\left\{ 0\right\} \right)$ and $\overline{\ddot{P}}\left(\left\{ 0\right\} \right)$. 

The next result is important for multiple hypothesis testing (§\ref{sec:Statistical-inference})
and, more generally, for combining probabilities of independent events
rather than entire distributions. 
\begin{cor}
\label{cor:discrete}Let $\xi=\left\langle \xi_{1},\dots,\xi_{N}\right\rangle $,
where $\xi_{j}$ is a Bernoulli random variable and $\xi_{j}$ is
independent of $\xi_{J}$ for all $j,J=1,\dots N$. (The Bernoulli
distributions need not be identical: in general, each has a different
probability $\ddot{P}\left(\xi_{i}=0\right)$ of failure. Every $\ddot{P}\in\ddot{\mathcal{P}}$
 has a one-to-one correspondence to a tuple $\left\langle \ddot{P}\left(\xi_{1}=0\right),\dots,\ddot{P}\left(\xi_{N}=0\right)\right\rangle $.)
Assuming $\dot{\mathcal{P}}\cap\ddot{\mathcal{P}}$ is nonempty and
finite, let $\ddot{P}_{i}$ denote the $i$th of the $\nu$ members
of $\extreme\left(\dot{\mathcal{P}}\cap\ddot{\mathcal{P}}\right)=\left\{ \ddot{P}_{1},\dots,\ddot{P}_{\nu}\right\} $.
If the constraints are in the form of lower and upper probabilities
$\underline{P}_{0,1},\dots,\underline{P}_{0,N}$ and $\overline{P}_{0,1},\dots,\overline{P}_{0,N}$
such that
\begin{equation}
\dot{\mathcal{P}}=\left\{ \dot{P}\in\mathcal{P}:\underline{P}{}_{0j}\le\dot{P}\left(\xi_{j}=0\right)\le\overline{P}_{0j},j\in\left\{ 1,\dots,N\right\} \right\} ,\label{eq:discrete-plausible}
\end{equation}
then the set of combining distributions that satisfy the constraints
is 
\begin{equation}
\dot{\mathcal{P}}\cap\ddot{\mathcal{P}}=\left\{ \ddot{P}\in\ddot{\mathcal{P}}:\underline{P}{}_{0j}\le\ddot{P}\left(\xi_{j}=0\right)\le\overline{P}_{0j},j\in\left\{ 1,\dots,N\right\} \right\} .\label{eq:intersection}
\end{equation}
Further, $P^{+}=P^{\mathbf{w}_{\extreme\left(\dot{\mathcal{P}}\cap\ddot{\mathcal{P}}\right)}}$
is the combination of the distributions in $\ddot{\mathcal{P}}$ with
truth constrained by $\dot{\mathcal{P}}$, where 

\begin{equation}
\mathbf{w}_{\extreme\left(\dot{\mathcal{P}}\cap\ddot{\mathcal{P}}\right)}=\arg\sup_{\left\langle w_{1},\dots,w_{\nu}\right\rangle \in\mathfrak{W}}\sum_{i=1}^{\nu}w_{i}\sum_{j=1}^{N}\sum_{k=0}^{1}\ddot{P}_{i}\left(\xi_{j}=k\right)\log\frac{\ddot{P}_{i}\left(\xi_{j}=k\right)}{P^{\left\langle w_{1},\dots,w_{\nu}\right\rangle }\left(\xi_{j}=k\right)}\label{eq:discrete-weight}
\end{equation}
with the supremum over $\mathfrak{W}=\left\{ \left\langle w_{1}^{\prime},\dots,w_{\nu}^{\prime}\right\rangle \in\left(0,1\right]^{\nu}:\sum_{i=1}^{\nu}w_{i}^{\prime}=1\right\} $\textup{. }\end{cor}
\begin{proof}
Eq. \eqref{eq:intersection} is obvious from eq. \eqref{eq:discrete-plausible}.
By the independence condition, the chain rule for information divergence
\citep[see, e.g.,][Theorem 2.5.3]{CoverThomas1991} reduces finding
the weighting distribution for according to Theorem \ref{thm:extreme}
to finding 
\[
\mathbf{w}_{\extreme\left(\dot{\mathcal{P}}\cap\ddot{\mathcal{P}}\right)}=\arg\sup_{\left\langle w_{1},\dots,w_{\nu}\right\rangle \in\mathfrak{W}}\sum_{i=1}^{\nu}w_{i}\sum_{j=1}^{N}D\left(\ddot{P}_{i}\left(\xi_{j}=\bullet\right)||P^{\left\langle w_{1},\dots,w_{\nu}\right\rangle }\left(\xi_{j}=\bullet\right)\right),
\]
which, with eq. \eqref{eq:cross-entropy}, yields eq. \eqref{eq:discrete-weight}.
\end{proof}

\section{\label{sec:Statistical-inference}Distribution combination for statistical
inference}

Whereas much of the literature focuses on combining priors from experts,
Sec. \ref{sub:Combining-posteriors} instead focuses on combining
posteriors. The posterior-inference setting enables the combination
not only of Bayesian posterior distributions but also of confidence
intervals and p-values encoded as frequentist posterior distributions,
as will be explained in Sec. \ref{sub:Frequentist-posteriors}.

\subsection{\label{sub:Combining-posteriors}Combining posterior distributions
and probabilities}

In the context of posterior statistical inference, $\xi$ represents
a random parameter. Further, all distributions in $\dot{\mathcal{P}}$,
the set of plausible distributions of $\xi$, and $\ddot{\mathcal{P}}$,
the set of combining distributions of $\xi$, are posterior with respect
to the same data set $x$. All distributions in $\dot{\mathcal{P}}$
and all Bayesian posteriors in $\ddot{\mathcal{P}}$ are conditional
on $X=x$, where, for any such posterior, the distribution of $X$
depends on the random value of the parameter drawn from some prior.
$\ddot{\mathcal{P}}$ may also contain non-Bayesian posteriors such
as a confidence posterior or a distribution derived from a confidence
posterior (§\ref{sub:Frequentist-posteriors}).

Accordingly, the information divergence $D\left(P||Q\right)$ becomes
the amount of information that would be gained by replacing any posterior
$Q$ with the true posterior $P$. That interpretation leads to viewing
$D\left(P^{\prime}||P^{\prime\prime}\rightsquigarrow Q\right)$ as
the amount of information gained for statistical inference by using
some posterior $Q\in\mathcal{P}$ rather than a given posterior $P^{\prime\prime}\in\ddot{\mathcal{P}}$
when the plausible posterior is $P^{\prime}\in\dot{\mathcal{P}}$.
Thus, $D\left(P^{\prime}||P^{\prime\prime}\rightsquigarrow Q\right)$
defines the \emph{inferential gain} \emph{of} $Q$ \emph{relative
to} $P^{\prime\prime}$ \emph{given} $P^{\prime}$ \citep{continuum,caution}. 

The posterior distributions are combined according to Sec. \ref{sub:Distribution-combination-game},
using Theorem \ref{thm:combination} whenever possible. If $\dot{\mathcal{P}}$
represents the uncertainty around a Bayesian posterior $\dot{P}\in\dot{\mathcal{P}}$,
as in \citet{RefWorks:1618} and \citet{caution}, then $\dot{P}$
is included in $\ddot{\mathcal{P}}$ as one of the distributions to
combine. The resulting combination $P^{+}$ is then used to minimize
expected loss in order to optimize actions such as point, interval,
and function estimators and predictors.

In model selection and hypothesis testing, $\xi$ has 0 or 1 as its
realized value, with $\xi=0$ if a reduced model or null hypothesis
is true or $\xi=1$ if a full model or alternative hypothesis is true.
Corollary \ref{cor:combining-probabilities} applies to this problem
with $\dot{\mathcal{P}}_{0}$ as the set of feasible null hypothesis
posterior probabilities and $\ddot{\mathcal{P}}_{0}=\left\{ \ddot{P}_{1}\left(\left\{ 0\right\} \right),...,\ddot{P}_{c}\left(\left\{ 0\right\} \right)\right\} $
as the set of null hypothesis posterior probabilities to be combined,
where $\ddot{P}_{i}\left(\left\{ 0\right\} \right)=\ddot{P}_{i}\left(\xi=0\right)$
is the $i$th posterior probability that the null hypothesis is true.
Thus, the combination posterior probability that the null hypothesis
is true is
\begin{equation}
P^{+}\left(\xi=0\right)=w^{+}\underline{\ddot{P}}\left(\xi=0\right)+\left(1-w^{+}\right)\overline{\ddot{P}}\left(\xi=0\right).\label{eq:combined-hypothesis-probability}
\end{equation}
Here, $\underline{\ddot{P}}\left(\xi=0\right)=\underline{\ddot{P}}\left(0\right)$
and $\overline{\ddot{P}}\left(\xi=0\right)=\overline{\ddot{P}}\left(0\right)$
are respectively the lowest and highest null hypothesis posterior
probabilities that are in $\dot{\mathcal{P}}_{0}\cap\ddot{\mathcal{P}}_{0}$,
presently assumed to have at least one member, and $w^{+}$ is determined
by eq. \eqref{eq:probability-weight}. The same idea applies to multiple
hypothesis testing, as will be seen in Sec. \ref{sec:Large-scale-case-study}.

\subsection{\label{sub:Frequentist-posteriors}Combining frequentist posteriors}

\subsubsection{Confidence posteriors}

As mentioned in Sec. \ref{sub:Frequentist-posteriors}, the set $\ddot{\mathcal{P}}$
of posterior combining distributions can include those representing
confidence intervals and p-values. To emphasize their comparability
to Bayesian posterior distributions, these frequentist distributions
are called {}``confidence posterior distributions'' \citep{CoherentFrequentism},
also known as {}``confidence distributions'' \citep[see, e.g.,][]{RefWorks:127}. 

Briefly, a \emph{confidence posterior distribution} that corresponds
to a set of nested confidence intervals evaluated for the observed
data is defined as the probability distribution according to which
the posterior probability that the interest parameter lies within
a confidence interval is equal to the confidence level of the interval.
For example, if a 95\% confidence interval for a real parameter is
$\left[-2.2,1.7\right]$, then there is a 95\% posterior probability
that the parameter is between $-2.2$ and $1.7$ according to the
confidence posterior. The same confidence posterior for the data also
assigns posterior probability to parameter intervals of interest according
to the confidence levels of the matching confidence intervals, e.g.,
\citet{conditional2009} considered a one-sided p-value as the posterior
probability that the population mean is in $\left(-\infty,0\right)$
rather than $\left[0,\infty\right)$. \citet{RefWorks:249} and \citet{Polansky2007b}
considered exact confidence posterior probabilities of intervals or
other regions specified before observing the data as ideal cases of
{}``attained confidence levels'' and {}``observed confidence levels,''
respectively.

\citet{CoherentFrequentism,conditional2009} proposed taking actions
that minimize expected loss with respect to a confidence posterior
distribution. Since that distribution is a Kolmogorov probability
distribution of the parameter of interest, such actions comply with
most axiomatic systems usually considered Bayesian, e.g., the systems
of \citet{MaxUtility1944} and \citet{RefWorks:126}. A human or artificial
intelligent agent that bets and makes other decisions in accordance
with minimizing expected loss with respect to a confidence posterior
corresponds to equating the confidence level of a confidence interval
with the agent's level of belief that the parameter value lies in
the interval \citep{CoherentFrequentism}. 

The decision-theoretic framework makes confidence posteriors suitable
as members of $\ddot{\mathcal{P}}$, the set of combining distributions,
according to the methodology of Sec. \ref{sec:Framework-for-combining}.
They can be combined to not only with each other, but also with other
parameter distributions such as Bayesian posteriors based on proper
or improper priors. The same applies to a probability distribution
of a function of a parameter drawn from a confidence posterior. Such
posteriors have been used to equate posterior probabilities of simple
null hypotheses with two-sided p-values \citep{RefWorks:1369,smallScale,continuum,caution}.
For terminological economy, these posteriors will now be called {}``confidence
posteriors.''

To approximate Bayesian model averaging, \citet{RefWorks:1016} recommended
a weighted harmonic mean of p-values computed from the same data,
provided that they range from $10^{-3}$ to $0.2$, the limits used
in Fig. \ref{fig:means}. Since the one-sided or two-sided p-values
are posterior probabilities of the null hypothesis derived from different
confidence posteriors, eqs. \eqref{eq:probability-weight} and \eqref{eq:combined-hypothesis-probability}
can be applied with $\underline{\ddot{P}}\left(\xi=0\right)$ and
$\overline{\ddot{P}}\left(\xi=0\right)$ as the lowest and highest
p-values that are plausible as null hypothesis probabilities, i.e.,
that are in $\dot{\mathcal{P}}_{0}$. The resulting combination p-value
differs from that of \citet{RefWorks:1016} in two respects: the mean
is arithmetic \eqref{eq:combined-hypothesis-probability} rather than
harmonic and, even more important, the weights are optimal for the
game \eqref{eq:probability-weight} rather than subjective. The use
of optimal weights leads to preparing for the worst case by averaging
only the two most extreme p-values rather than all of them. 
\begin{example}
Given a small sample of data drawn from a distribution that might
be approximately normal, let $p^{\left(1\right)}$ and $p^{\left(2\right)}$
denote the two-sided p-values according to the \emph{t}-test and the
Wilcoxon signed-rank test, respectively; $p^{\left(1\right)}<p^{\left(2\right)}$.
Under conditions often applicable to simple (point) hypothesis testing
with a diffuse alternative hypothesis  \citep{RefWorks:1218}, the
plausible set of posterior probabilities of the null hypothesis is
$\dot{\mathcal{P}}_{0}=\left[\underline{\dot{P}}_{0},1\right]$ with
lower bound
\[
\underline{\dot{P}}_{0}=\left(1+\left(\frac{1-\underline{\dot{P}}_{0}^{\pri}}{\underline{\dot{P}}_{0}^{\pri}ep^{\left(2\right)}\left(x\right)\log\left[1/p^{\left(2\right)}\left(x\right)\right]}\right)\right)^{-1}\wedge\underline{\dot{P}}_{0}^{\pri},
\]
where $\wedge$ is the minimum operator, and $\underline{\dot{P}}_{0}^{\pri}$
is the lowest plausible prior probability that the null hypothesis
is true \citep{caution}. Then the combined p-value $P^{+}\left(\left\{ 0\right\} \right)$
is $\underline{\dot{P}}_{0}$ if $p^{\left(2\right)}<\underline{\dot{P}}_{0}$,
$p^{\left(2\right)}$ if $p^{\left(1\right)}<\underline{\dot{P}}_{0}\le p^{\left(2\right)}$,
and, according to Corollary \ref{cor:combining-probabilities}, the
weighted arithmetic mean $w^{+}p^{\left(1\right)}+\left(1-w^{+}\right)p^{\left(2\right)}$
if $p^{\left(1\right)}\ge\underline{\dot{P}}_{0}$ with the weights
$w^{+}$ and $\left(1-w^{+}\right)$ fixed by eq. \eqref{eq:probability-weight}.
Of the three cases, the third yields a combined p-value that differs
from the blended posterior probability suggested in \citet{continuum}.
$\exaend$
\end{example}
When $\ddot{\mathcal{P}}$ consists of a single confidence posterior
$\ddot{P}$, the resulting $P^{+}$, degenerate as a {}``combination''
of a single distribution, is better viewed as a solution to the problem
of blending frequentist inference with constraints encoded as the
Bayesian posteriors that constitute $\dot{\mathcal{P}}$. That solution
in general differs from the minimax-type solutions considered \citep{continuum,caution}.
Under $\ddot{P}\in\dot{\mathcal{P}}$ and the convexity of $\dot{\mathcal{P}}$,
they lead to the $P$ that minimizes the information divergence $D\left(P||\ddot{P}\right)$,
which is dual to the $Q$ that minimizes $D\left(\ddot{P}||Q\right)$,
the information divergence that is minimized \eqref{eq:minKL} when
maximizing eq. \eqref{eq:utility-function} according to the game
introduced in Sec. \ref{sub:Distribution-combination-game}.

\subsubsection{\label{sub:Multiple-comparison-procedures}Multiple comparison procedures}

The distribution-combination theory is now applied to adjustments
for multiple comparisons by formalizing the observation that p-values
are often adjusted to the extent of prior belief in the null hypothesis.
That is, multiple comparison procedures (MCPs) designed to control
error rates are {}``most likely to be used, if at all, when most
of the individual null hypotheses are essentially correct'' \citep[p.88]{CoxBook}.
 A first-order formalization would take the p-value adjusted according
to an MCP as the posterior probability of the null hypothesis. To
the extent that the knowledge or opinion of the agent is such that
its decisions would be made to minimize the expected loss with respect
to that posterior distribution, the use of the MCP is warranted. In
this interpretation, combining p-values across different MCPs is equivalent
to combining the posterior distributions that represent the corresponding
opinions.
\begin{example}
The Bonferroni procedure controls the family-wise error rate, the
probability that one or more true null hypotheses will be rejected,
at any level $\alpha\in\left[0,1\right]$. That is accomplished on
the basis of p-values $p_{1},\dots,p_{N}$ by rejecting the $i$th
of $N$ null hypotheses if the adjusted p-value $Np_{i}\wedge1$ is
less than $\alpha$. Thus, the posterior probabilities generated by
the Bonferroni procedure are appropriate only when the prior probability
of each null hypothesis is inversely proportional to the number of
tests. As \citet{RefWorks:195} pointed out, the {}``Bonferroni method
is based upon the implicit presumption of a moderate degree of belief
in the event'' that all null hypotheses under consideration are true
and that the prior truth values of the hypotheses are approximately
independent. 

Accordingly, the Bonferroni method is widely used to analyze genome-wide
association data, largely because only an extremely small fraction
$\pi_{1}$ of the hundreds of thousands of markers tested are thought
to be associated with the trait of interest. \citet{RefWorks:199}
guessed $10^{-6}\le\pi_{1}\le10^{-4}$, interpreting $\pi_{1}$ as
the prior probability of association between a given marker and the
trait. The corresponding range of posterior probabilities and Bayes
factors such as those of \citet{RefWorks:199} would define $\dot{\mathcal{P}}$
for ruling out MCPs that yield implausible results (Theorem \ref{thm:extreme}).
(On the other hand, some evidence that $\pi_{1}\ge10^{-4}$ is now
available in preliminary estimates \citep{GWAselect} and in indications
that thousands of small-effect SNPs may be associated with any particular
disease \citep{Gibson2010,Park2010}.) $\exaend$
\end{example}
By assuming adjusted p-values are equal to independent posterior
probabilities of the null hypotheses, the methodology of Corollary
\ref{cor:discrete} can combine the results of various MCPs.

\section{\label{sec:Large-scale-case-study}Large-scale case study}

Using microarray technology, \citet{RefWorks:8} measured the levels
of tomato gene expression for 13,440 genes at three days after the
breaker stage of ripening, but one or more measurements were missing
for 7337 genes. The data available across all $n=6$ biological replicates
for $N=6103$ of the genes illustrate the methodology of Secs. \ref{sec:Framework-for-combining}
and \ref{sec:Statistical-inference}. 

For $j=1,\dots,N$, the logarithms of the measured ratios of mutant
expression to wild-type expression in the $j$th gene were modeled
as realizations of a normal variate and are denoted by the $n$-tuple
$x_{j}$. Because the mean and variance are unknown, the one-sample
\emph{t}-test was used to test the null hypothesis $\left(\xi_{j}=0\right)$
that the population mean is 0 against the two-sided alternative hypothesis
$\left(\xi_{j}=1\right)$ that there is differential expression of
the $j$th gene between mutant and wild type, i.e., that the mutation
affects the expression of gene $j$.

The posterior probability of a null hypothesis conditional on the
p-value is called its \emph{local false discovery rate} (LFDR) \citep{RefWorks:53}.
Three very different methods $\left(i=1,2,3\right)$ of estimating
the LFDR were considered. The first two methods are based on fitting
a histogram of transformed p-values that is described by \citet{RefWorks:208}.
They differ in that whereas one assumes the p-value has a uniform
distribution under the null hypothesis $\left(i=1\right)$, the other
estimates the p-value null distribution by maximizing a truncated
likelihood function $\left(i=2\right)$. Each method has its own advantages
(§\ref{sec:Introduction}). The distributions are called the \emph{theoretical
null} and the \emph{empirical null}, respectively. The third method
for combination is the q-value \citep{RefWorks:282}, here defined
according to the algorithm of \citet{RefWorks:288} as the lowest
false discovery rate at which a null hypothesis will be rejected $\left(i=3\right)$.
While the q-value was not originally intended as an estimator of the
LFDR, it is included here since its negative bias as such an estimator
\citep{ISI:000272935000021} may have a corrective effect on the positive
bias (conservatism) of the first two LFDR estimators.

For this application, $\mathcal{P}$ is the set of all probability
distributions on $\left(\left\{ 0,1\right\} ^{N},2^{\left\{ 0,1\right\} ^{N}}\right)$.
Corresponding to those three methods, let $\ddot{P}_{1}$, $\ddot{P}_{2}$,
and $\ddot{P}_{3}$ denote the members of $\mathcal{P}$ such that
the $i$th estimate of the LFDR of the $j$th gene is $\ddot{P}_{i}\left(\xi_{j}=0\right)$.
To combine the three methods, $\ddot{\mathcal{P}}=\left\{ \ddot{P}_{1},\ddot{P}_{2},\ddot{P}_{3}\right\} $
is taken as the set of $\nu=3$ combining distributions.

For the $j$th gene, $f\left(t\left(x_{j}\right);\theta_{j}\right)$
will represent the probability density of the Student $t$ statistic
$t\left(x_{j}\right)$, where $\theta_{j}$ is the reciprocal of the
coefficient of variation and, for any $\theta\in\mathbb{R}$, $f\left(\bullet;\theta\right)$
is the probability density function of $\left|T\right|$ when $T$
has the noncentral $t$ distribution of $n-1$ degrees of freedom
and noncentrality parameter $\sqrt{n}\theta$. The set $\dot{\mathcal{P}}$
of plausible distributions will be determined on the basis of $\left\{ L_{j}\left(\bullet\right)=f\left(t\left(x_{j}\right);\bullet\right):j=1,\dots,N\right\} $,
the set of likelihood functions. The plausible distributions are also
based on $\underline{\pi}_{0}=80\%$, an assumed lower bound on the
proportion of genes that are not differentially expressed. By Bayes's
theorem, the posterior odds of the $j$th null hypothesis is the product
of the prior odds, which is the least $\underline{\pi}_{0}/\left(1-\underline{\pi}_{0}\right)$,
and the Bayes factor, which must be at least $L_{j}\left(0\right)/\max_{\theta\ne0}L_{j}\left(\theta\right)$.
Thus, for gene $j$, a lower bound $\underline{\Omega}_{j}$ of the
posterior odds is the product of the last two quantities, and a lower
bound of the LFDR is $\underline{P}\left(\xi_{j}=0\right)=\underline{\Omega}_{j}/\left(1+\underline{\Omega}_{j}\right)$.
In the notation of Corollary \ref{cor:discrete}, $\underline{P}_{0,j}=\underline{P}\left(\xi_{j}=0\right)$
and, trivially, $\overline{P}_{0,j}=1$ for all $j=1,\dots,N$. Thus,
the plausible set specified by eq. \eqref{eq:discrete-plausible}
consists of the posterior distributions satisfying the lower bound
derived from the likelihood functions and $\underline{\pi}_{0}=80\%$.

The horizontal axis and straight line in Fig. \ref{fig:Combination-xprn}
represent $\underline{P}$, and the intermediate, highest, and lowest
dashed curves represent $\ddot{P}_{1}$, $\ddot{P}_{2}$, and $\ddot{P}_{3}$,
respectively. Since some of the q-values are less than the lower bound
$\left(\exists j:\ddot{P}_{3}\left(\xi_{j}=0\right)<\underline{P}\left(\xi_{j}=0\right)\right)$
but all of the other LFDR estimates satisfy the bound $\left(i=1,2;\forall j:\ddot{P}_{i}\left(\xi_{j}=0\right)\ge\underline{P}\left(\xi_{j}=0\right)\right)$,
the former are excluded when computing the combined estimates according
to eq. \eqref{eq:intersection}, in which $\dot{\mathcal{P}}\cap\ddot{\mathcal{P}}=\left\{ \ddot{P}_{1},\ddot{P}_{2}\right\} $.
Since there are only two distributions, each corresponds to an extreme
point, leading to $\extreme\left(\dot{\mathcal{P}}\cap\ddot{\mathcal{P}}\right)=\left\{ \ddot{P}_{1},\ddot{P}_{2}\right\} $
in eq. \eqref{eq:discrete-weight}. Thereby, the combined distribution
is numerically found to be the linear combination $P^{+}=w_{1}\ddot{P}_{1}+w_{2}\ddot{P}_{2}$
with $w_{1}=0.43$ and $w_{2}=0.57$. By implication, the game-optimal
LFDR estimate for the $j$th gene is $P^{+}\left(\xi_{j}=0\right)=w_{1}\ddot{P}_{1}\left(\xi_{j}=0\right)+w_{2}\ddot{P}_{2}\left(\xi_{j}=0\right)$.
Those combined estimates are plotted as the solid curve in Fig. \ref{fig:Combination-xprn}.
\section*{Acknowledgments}

I thank Xuemei Tang for sending me the fruit development microarray
data. This research was partially supported  by the Canada Foundation
for Innovation, by the Ministry of Research and Innovation of Ontario,
and by the Faculty of Medicine of the University of Ottawa. 

\begin{flushleft}
\bibliographystyle{elsarticle-harv}
\bibliography{refman}

\begin{thebibliography}{57}
\expandafter\ifx\csname natexlab\endcsname\relax\def\natexlab#1{#1}\fi
\expandafter\ifx\csname url\endcsname\relax
  \def\url#1{\texttt{#1}}\fi
\expandafter\ifx\csname urlprefix\endcsname\relax\def\urlprefix{URL }\fi

\bibitem[{Abbas(2009)}]{Abbas2009}
Abbas, A.~E., Mar. 2009. {A Kullback-Leibler View of Linear and Log-Linear
  Pools}. Decision Analysis 6, 25--37.

\bibitem[{Alba et~al.(2005)Alba, Payton, Fei, McQuinn, Debbie, Martin,
  Tanksley, and Giovannoni}]{RefWorks:8}
Alba, R., Payton, P., Fei, Z., McQuinn, R., Debbie, P., Martin, G.~B.,
  Tanksley, S.~D., Giovannoni, J.~J., 2005. Transcriptome and selected
  metabolite analyses reveal multiple points of ethylene control during tomato
  fruit development. Plant Cell 17, 2954--2965.

\bibitem[{Barndorff-Nielsen(1995)}]{DiversityEvidence1995}
Barndorff-Nielsen, O.~E., 1995. Diversity of evidence and birnbaum's theorem.
  Scandinavian Journal of Statistics 22, 513--515.

\bibitem[{Benjamini and Hochberg(1995)}]{RefWorks:288}
Benjamini, Y., Hochberg, Y., 1995. Controlling the false discovery rate: A
  practical and powerful approach to multiple testing. Journal of the Royal
  Statistical Society B 57, 289--300.

\bibitem[{Berger(1985)}]{RefWorks:179}
Berger, J.~O., 1985. Statistical Decision Theory and {B}ayesian Analysis.
  Springer, New York.

\bibitem[{Bickel(2011{\natexlab{a}})}]{continuum}
Bickel, D.~R., 2011{\natexlab{a}}. Blending {B}ayesian and frequentist methods
  according to the precision of prior information with an application to
  hypothesis testing. Technical Report, Ottawa Institute of Systems Biology,
  arXiv:1107.2353.

\bibitem[{Bickel(2011{\natexlab{b}})}]{CoherentFrequentism}
Bickel, D.~R., 2011{\natexlab{b}}. Coherent frequentism: A decision theory
  based on confidence sets. To appear in Communications in Statistics - Theory
  and Methods (accepted 22 November 2010); 2009 preprint available from
  arXiv:0907.0139.

\bibitem[{Bickel(2011{\natexlab{c}})}]{caution}
Bickel, D.~R., 2011{\natexlab{c}}. Controlling the degree of caution in
  statistical inference with the {B}ayesian and frequentist approaches as
  opposite extremes. Technical Report, Ottawa Institute of Systems Biology,
  arXiv:1109.5278.

\bibitem[{Bickel(2011{\natexlab{d}})}]{conditional2009}
Bickel, D.~R., 2011{\natexlab{d}}. Estimating the null distribution to adjust
  observed confidence levels for genome-scale screening. Biometrics 67,
  363--370.

\bibitem[{Bickel(2011{\natexlab{e}})}]{smallScale}
Bickel, D.~R., 2011{\natexlab{e}}. Small-scale inference: Empirical {B}ayes and
  confidence methods for as few as a single comparison. Technical Report,
  Ottawa Institute of Systems Biology, arXiv:1104.0341.

\bibitem[{Ciesielski(1997)}]{Ciesielski1997}
Ciesielski, K., 1997. {Set Theory for the Working Mathematician}. Cambridge
  University Press, Cambridge.

\bibitem[{Clemen and Winkler(1999)}]{RefWorks:521}
Clemen, R.~T., Winkler, R.~L., 1999. Combining probability distributions from
  experts in risk analysis. Risk Analysis 19, 187--203.

\bibitem[{Cooke(1991)}]{cooke_experts_1995}
Cooke, R.~M., 1991. Experts in Uncertainty: Opinion and Subjective Probability
  in Science. Oxford University Press.

\bibitem[{Cover and Thomas(2006)}]{CoverThomas1991}
Cover, T., Thomas, J., 2006. Elements of Information Theory. John Wiley and
  Sons, New York.

\bibitem[{Cox(2006)}]{CoxBook}
Cox, D.~R., 2006. Principles of Statistical Inference. Cambridge University
  Press, Cambridge.

\bibitem[{Csisz\'{a}r and K\"{o}rner(2011)}]{Csiszar2011}
Csisz\'{a}r, I., K\"{o}rner, J., 2011. {Information Theory: Coding Theorems for
  Discrete Memoryless Systems}. Cambridge University Press, Cambridge.

\bibitem[{Davisson and Leon-Garcia(1980)}]{Davisson1980}
Davisson, L., Leon-Garcia, a., Mar. 1980. {A source matching approach to
  finding minimax codes}. IEEE Transactions on Information Theory 26, 166--174.

\bibitem[{Efron(2007)}]{RefWorks:208}
Efron, B., 2007. Size, power and false discovery rates. Annals of Statistics
  35, 1351--1377.

\bibitem[{Efron(2010)}]{efron_large-scale_2010}
Efron, B., 2010. Large-Scale Inference: Empirical {B}ayes Methods for
  Estimation, Testing, and Prediction. Cambridge University Press.

\bibitem[{Efron and Tibshirani(1998)}]{RefWorks:249}
Efron, B., Tibshirani, R., 1998. The problem of regions. Annals of Statistics
  26, 1687--1718.

\bibitem[{Efron et~al.(2001)Efron, Tibshirani, Storey, and
  Tusher}]{RefWorks:53}
Efron, B., Tibshirani, R., Storey, J.~D., Tusher, V., 2001. Empirical {B}ayes
  analysis of a microarray experiment. J. Am. Stat. Assoc. 96, 1151--1160.

\bibitem[{Gajdos et~al.(2004)Gajdos, Tallon, and Vergnaud}]{RefWorks:1618}
Gajdos, T., Tallon, J.~M., Vergnaud, J.~C., SEP 2004 2004. Decision making with
  imprecise probabilistic information. Journal of Mathematical Economics 40,
  647--681.

\bibitem[{Gallager(1979)}]{Gallager1979}
Gallager, R.~G., 1979. {Source coding with side information and universal
  coding}. Technical Report LIDS-P-937, Laboratory for Information Decision
  Systems, MIT.

\bibitem[{Genest and Zidek(1986)}]{Genest1986}
Genest, C., Zidek, J.~V., 1986. {Combining Probability Distributions: A
  Critique and an Annotated Bibliography}. Statistical Science 1, 114--135.

\bibitem[{Gibson(2010)}]{Gibson2010}
Gibson, G., Jul. 2010. {Hints of hidden heritability in GWAS}. Nature Genetics
  42, 558--60.

\bibitem[{Good({1984})}]{ISI:A1984TE72700008}
Good, I., {1984}. {A {B}ayesian interpretation of ancillarity}. {Journal of
  Statistical Computation and Simulation} {19}~({4}), {302--308}.

\bibitem[{Good(1958)}]{RefWorks:1016}
Good, I.~J., 1958. Significance tests in parallel and in series. Journal of the
  American Statistical Association 53, 799--813.

\bibitem[{Gr\"{u}nwald and Philip~Dawid(2004)}]{Gruenwald20041367}
Gr\"{u}nwald, P., Philip~Dawid, A., 2004. Game theory, maximum entropy, minimum
  discrepancy and robust {B}ayesian decision theory. Annals of Statistics 32,
  1367--1433.

\bibitem[{Haussler(1997)}]{Haussler1997b}
Haussler, D., 1997. {A general minimax result for relative entropy}. IEEE
  Transactions on Information Theory 43, 1276 -- 1280.

\bibitem[{Hong et~al.({2009})Hong, Tibshirani, and Chu}]{ISI:000272935000021}
Hong, W.-J., Tibshirani, R., Chu, G., {2009}. {Local false discovery rate
  facilitates comparison of different microarray experiments}. {NUCLEIC ACIDS
  RESEARCH} {37}~({22}), {7483--7497}.

\bibitem[{Keeney and Raiffa(1993)}]{KeeneyRaiffa1993b}
Keeney, R.~L., Raiffa, H., 1993. {Decisions with Multiple Objectives:
  Preferences and Value Tradeoffs}. Cambridge University Press, Cambridge.

\bibitem[{Koshy(2004)}]{Koshy2004}
Koshy, T., 2004. {Discrete mathematics with applications}. Academic Press.

\bibitem[{Krac\'{\i}k(2011)}]{Kracik2011}
Krac\'{\i}k, J., 2011. {Combining marginal probability distributions via
  minimization of weighted sum of Kullback-Leibler divergences}. International
  Journal of Approximate Reasoning 52, 659--671.

\bibitem[{Levi(1986{\natexlab{a}})}]{Levi1986}
Levi, I., 1986{\natexlab{a}}. {Hard Choices: Decision Making under Unresolved
  Conflict}. Cambridge University Press, Cambridge.

\bibitem[{Levi(1986{\natexlab{b}})}]{Levi1986a}
Levi, I., 1986{\natexlab{b}}. {The paradoxes of Allais and Ellsberg}. Economics
  and Philosophy 2, 23--53.

\bibitem[{McConway(1981)}]{Mcconway1981}
McConway, K.~J., 1981. {Marginalization and linear opinion pools}. Journal of
  the American Statistical Association 76, 410--414.

\bibitem[{Nakagawa and Kanaya(1988)}]{Nakagawa1988}
Nakagawa, K., Kanaya, F., 1988. {A new geometric capacity characterization of a
  discrete memoryless channel}. IEEE Transactions on Information Theory 34,
  318--321.

\bibitem[{Park et~al.(2010)Park, Wacholder, Gail, Peters, Jacobs, Chanock, and
  Chatterjee}]{Park2010}
Park, J.-H., Wacholder, S., Gail, M.~H., Peters, U., Jacobs, K.~B., Chanock,
  S.~J., Chatterjee, N., Jul. 2010. {Estimation of effect size distribution
  from genome-wide association studies and implications for future
  discoveries}. Nature Genetics 42, 570--5.

\bibitem[{Pfaffelhuber(1977)}]{Phuber77}
Pfaffelhuber, E., 1977. {Minimax Information Gain and Minimum Discrimination
  Principle}. In: Csisz{\'a}r, I., Elias, P. (Eds.), Topics in Information
  Theory. Vol.~16 of Colloquia Mathematica Societatis J{\'a}nos Bolyai.
  J{\'a}nos Bolyai Mathematical Society and North-Holland, pp. 493--519.

\bibitem[{Polansky(2007)}]{Polansky2007b}
Polansky, A.~M., 2007. Observed Confidence Levels: Theory and Application.
  Chapman and Hall, New York.

\bibitem[{Rissanen(2007)}]{RefWorks:374}
Rissanen, J., 2007. Information and Complexity in Statistical Modeling.
  Springer, New York.

\bibitem[{Ryabko(1979)}]{Ryabko197971}
Ryabko, B., 1979. Encoding of a source with unknown but ordered probabilities.
  Prob. Pered. Inform. 15, 71--77.

\bibitem[{Ryabko(1981)}]{Ryabko1981}
Ryabko, B., Nov. 1981. {Comments on 'A source matching approach to finding
  minimax codes' by Davisson, L. D. and Leon-Garcia, A.} IEEE Transactions on
  Information Theory 27, 780--781.

\bibitem[{Savage(1954)}]{RefWorks:126}
Savage, L.~J., 1954. The Foundations of Statistics. John Wiley and Sons, New
  York.

\bibitem[{Schweder and Hjort(2002)}]{RefWorks:127}
Schweder, T., Hjort, N.~L., 2002. Confidence and likelihood. Scandinavian
  Journal of Statistics 29, 309--332.

\bibitem[{Seidenfeld(2004)}]{Seidenfeld2004}
Seidenfeld, T., 2004. {A contrast between two decision rules for use with
  (convex) sets of probabilities: $\Gamma$-maximin versus}. Synthese 140,
  69--88.

\bibitem[{Sellke et~al.(2001)Sellke, Bayarri, and Berger}]{RefWorks:1218}
Sellke, T., Bayarri, M.~J., Berger, J.~O., 2001. Calibration of p values for
  testing precise null hypotheses. American Statistician 55, 62--71.

\bibitem[{Shulman and Feder(2004)}]{Shulman2004}
Shulman, N., Feder, M., 2004. {The uniform distribution as a universal prior}.
  IEEE Transactions on Information Theory 50, 581--586.

\bibitem[{Storey(2002)}]{RefWorks:282}
Storey, J.~D., 2002. A direct approach to false discovery rates. Journal of the
  Royal Statistical Society. Series B: Statistical Methodology 64, 479--498.

\bibitem[{Toda(1956)}]{Toda1956}
Toda, M., May 1956. {Information-receiving behavior of man}. Psychological
  Review 63, 204--212.

\bibitem[{Tops{\o}e(2007)}]{Topsoe2007b}
Tops{\o}e, F., 2007. Information theory at the service of science. In: Tóth,
  G.~F., Katona, G. O.~H., Lovász, L., Pálfy, P.~P., Recski, A., Stipsicz, A.,
  Szász, D., Miklós, D., Csisz\'{a}r, I., Katona, G. O.~H., Tardos, G., Wiener,
  G. (Eds.), Entropy, Search, Complexity. Bolyai Society Mathematical Studies.
  Springer Berlin Heidelberg, pp. 179--207.

\bibitem[{van Berkum et~al.(1996)van Berkum, Linssen, and
  Overdijk}]{RefWorks:1369}
van Berkum, E., Linssen, H., Overdijk, D., 1996. Inference rules and
  inferential distributions. Journal of Statistical Planning and Inference 49,
  305--317.

\bibitem[{von Neumann and Morgenstern(1953)}]{MaxUtility1944}
von Neumann, J., Morgenstern, O., 1953. Theory of Games and Economic Behavior.
  Princeton University Press, Princeton.

\bibitem[{Wald(1961)}]{Wald1950}
Wald, A., 1961. Statistical Decision Functions. John Wiley and Sons, New York.

\bibitem[{{Wellcome Trust Case Control Consortium}(2007)}]{RefWorks:199}
{Wellcome Trust Case Control Consortium}, 2007. Genome-wide association study
  of 14,000 cases of seven common diseases and 3,000 shared controls. Nature
  447, 661--678.

\bibitem[{Westfall et~al.(1997)Westfall, Johnson, and Utts}]{RefWorks:195}
Westfall, P.~H., Johnson, W.~O., Utts, J.~M., 1997. A {B}ayesian perspective on
  the {B}onferroni adjustment. Biometrika 84, 419--427.

\bibitem[{Yang and Bickel(2010)}]{GWAselect}
Yang, Y., Bickel, D.~R., 2010. Minimum description length and empirical {B}ayes
  methods of identifying {SNPs} associated with disease. Technical Report,
  Ottawa Institute of Systems Biology, COBRA Preprint Series, Article 74,
  biostats.bepress.com/cobra/ps/art74.

\end{thebibliography}

\par\end{flushleft}

\newpage{}

\textbf{\LARGE ~}{\LARGE \par}

\textbf{\LARGE  }{\LARGE \par}

\section*{Appendix A: Remarks}

~
\begin{rem}
\label{rem:game-theory}(Sec. \ref{sec:Introduction}) Since formulating
the distribution combination problem in terms of a game is unconventional,
it is worth noting that game theory laid the foundations of the two
dominant schools of statistical decision theory. The maximum-expected-payoff
solution of a one-player game \citep[Ch. I]{MaxUtility1944} led to
axiomatic systems that support Bayesian statistics \citep[e.g.,][]{RefWorks:126}.
Likewise, the worst-case (minimax) solutions of certain two-player
zero-sum games \citep[Ch. III]{MaxUtility1944} led to frequentist
decision theory \citep{Wald1950}.
\end{rem}
~
\begin{rem}
\label{rem:minimax-redundancy}(Sec. \ref{sub:Information-theory})
The discrete-distribution version of Lemma \ref{lem:minimax-redundancy},
the main result of the {}``redundancy-capacity theorem,'' was presented
by R. G. Gallager in 1974 \citep[Editor's Note]{Ryabko1981} and published
by \citet{Ryabko197971} and \citet{Davisson1980}; cf. \citet{Gallager1979}.
\citet[Theorem 13.1.1]{CoverThomas1991}, \citet[§5.2.1]{RefWorks:374},
and \citet[Problem 8.1]{Csiszar2011} provide useful introductions.
\end{rem}
~
\begin{rem}
\label{rem:lexicographic}(Sec. \ref{sub:Distribution-combination-game})
Previous instances of lexicographically maximizing expected utility
with respect to an optimal probability distribution include the use
of the least informative prior \citep{Seidenfeld2004} and the use
of the posterior $P_{2}$ used to maximize $D\left(\dot{P}||P_{1}\rightsquigarrow P_{2}\right)$
in a two-player zero-sum game \citep{caution}. On lexicographic decision
making in other contexts, see \citet[§§5.7, 6.9]{Levi1986}, \citet{Levi1986a},
and \citet[§3.3.1]{KeeneyRaiffa1993b}. \citet[Ch. 4]{Ciesielski1997}
and \citet[Ch. 7]{Koshy2004} provide more formal set-theoretic expositions
of lexicographic ordering.
\end{rem}

\section*{Appendix B: Additional proofs}

\subsection*{Proof of Theorem \ref{thm:combination}}

For any $Q\in\mathcal{P}$, eqs. \eqref{eq:information-gain} and
\eqref{eq:utility-function} yield
\begin{eqnarray*}
\sup_{\left\langle P^{\prime},P^{\prime\prime}\right\rangle \in\dot{\mathcal{P}}\times\ddot{\mathcal{P}}}^{\preceq}U\left(P^{\prime},Q,P^{\prime\prime}\right) & = & \sup_{\left\langle P^{\prime},P^{\prime\prime}\right\rangle \in\dot{\mathcal{P}}\times\ddot{\mathcal{P}}}^{\preceq}\left\langle -D\left(P^{\prime}||P^{\prime\prime}\right),D\left(P^{\prime}||Q\rightsquigarrow P^{\prime\prime}\right)\right\rangle \\
 & = & \sup_{P^{\prime}\in\dot{\mathcal{P}}\cap\ddot{\mathcal{P}}}^{\preceq}\left\langle -D\left(P^{\prime}||P^{\prime}\right),D\left(P^{\prime}||Q\rightsquigarrow P^{\prime}\right)\right\rangle \\
 & = & \sup_{P^{\prime}\in\dot{\mathcal{P}}\cap\ddot{\mathcal{P}}}\left[D\left(P^{\prime}||Q\right)-D\left(P^{\prime}||P^{\prime}\right)\right]=\sup_{P^{\prime}\in\dot{\mathcal{P}}\cap\ddot{\mathcal{P}}}D\left(P^{\prime}||Q\right).
\end{eqnarray*}
Thus, by eqs. \eqref{eq:information-gain}, \eqref{eq:utility-function},
\eqref{eq:combined-distribution}, and \eqref{eq:opponent-moves},
\begin{eqnarray}
P^{+} & = & \arg\sup_{Q\in\mathcal{P}}^{\preceq}U\left(\arg\sup_{P^{\prime}\in\dot{\mathcal{P}}\cap\ddot{\mathcal{P}}}D\left(P^{\prime}||Q\right),\arg\sup_{P^{\prime}\in\dot{\mathcal{P}}\cap\ddot{\mathcal{P}}}D\left(P^{\prime}||Q\right),Q\right)\nonumber \\
 & = & \arg\sup_{Q\in\mathcal{P}}\left(-D\left(\arg\sup_{P^{\prime}\in\dot{\mathcal{P}}\cap\ddot{\mathcal{P}}}D\left(P^{\prime}||Q\right)||Q\right)\right)\nonumber \\
 & = & \arg\inf_{Q\in\mathcal{P}}\sup_{P^{\prime}\in\dot{\mathcal{P}}\cap\ddot{\mathcal{P}}}D\left(P^{\prime}||Q\right).\label{eq:minKL}
\end{eqnarray}
Hence, according to eq. \eqref{eq:centroid}, $P^{+}$ is $\widetilde{\dot{\mathcal{P}}\cap\ddot{\mathcal{P}}}$,
the centroid of $\dot{\mathcal{P}}\cap\ddot{\mathcal{P}}$. Since
$\dot{\mathcal{P}}\cap\ddot{\mathcal{P}}\ne\emptyset$ by assumption,
the conditions of Lemma \ref{lem:minimax-redundancy} are satisfied.

\subsection*{Proof of Lemma \ref{lem:extreme-set}}

As an immediate consequence of what \citet{Nakagawa1988} label {}``Theorem
(Csiszár)'' and {}``Theorem 1,'' 
\[
\min_{P^{\prime\prime}\in\mathcal{P}}\max_{P^{\prime}\in\mathcal{P}^{\star}}D\left(P^{\prime}||P^{\prime\prime}\right)=C.
\]
By definition, the centroid is the solution of that minimax problem
\eqref{eq:centroid}.

\subsection*{Proof of Theorem \ref{thm:extreme}}

Lemma \ref{lem:minimax-redundancy} implies that $\widetilde{\extreme\left(\dot{\mathcal{P}}\cap\ddot{\mathcal{P}}\right)}$,
the centroid of $\extreme\left(\dot{\mathcal{P}}\cap\ddot{\mathcal{P}}\right)$,
is $P^{W_{\extreme\left(\dot{\mathcal{P}}\cap\ddot{\mathcal{P}}\right)}}$.
Since, according to the definition of an extreme point and the definition
of a centroid \eqref{eq:centroid}, it is not possible that there
exist a $P^{\prime}\in\extreme\left(\dot{\mathcal{P}}\cap\ddot{\mathcal{P}}\right)$
and a $P^{\prime\prime}\in\extreme\left(\dot{\mathcal{P}}\cap\ddot{\mathcal{P}}\right)$
such that $D\left(P^{\prime}||\widetilde{\extreme\left(\dot{\mathcal{P}}\cap\ddot{\mathcal{P}}\right)}\right)<D\left(P^{\prime\prime}||\widetilde{\extreme\left(\dot{\mathcal{P}}\cap\ddot{\mathcal{P}}\right)}\right)$,
it follows that 
\[
D\left(P^{\star}||P^{W_{\extreme\left(\dot{\mathcal{P}}\cap\ddot{\mathcal{P}}\right)}}\right)=\max_{P^{\prime}\in\extreme\left(\dot{\mathcal{P}}\cap\ddot{\mathcal{P}}\right)}D\left(P^{\prime}||P^{W_{\extreme\left(\dot{\mathcal{P}}\cap\ddot{\mathcal{P}}\right)}}\right)
\]
for all $P^{\star}\in\extreme\left(\dot{\mathcal{P}}\cap\ddot{\mathcal{P}}\right)$.
According to Lemma \ref{lem:extreme-set}, $P^{W_{\extreme\left(\dot{\mathcal{P}}\cap\ddot{\mathcal{P}}\right)}}$
is $\widetilde{\dot{\mathcal{P}}\cap\ddot{\mathcal{P}}}$, the centroid
of $\dot{\mathcal{P}}\cap\ddot{\mathcal{P}}$. That centroid is $P^{+}$
by Theorem \ref{thm:combination}.

\subsection*{}

\newpage{}
\begin{figure}
\includegraphics{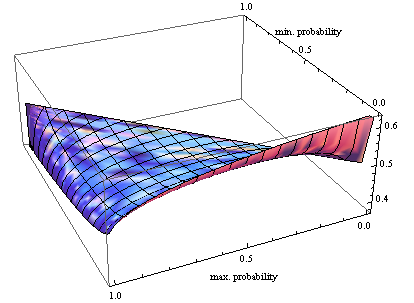}

\caption{\label{fig:weight}Optimal weight $w^{+}$ versus $\underline{\ddot{P}}\left(\left\{ 0\right\} \right)$
and $\overline{\ddot{P}}\left(\left\{ 0\right\} \right)$, the lowest
and highest of the plausible probabilities to be combined \eqref{eq:probability-weight},
labeled here as {}``min. probability'' and {}``max. probability,''
respectively. The combination probability is $P^{+}\left(0\right)=w^{+}\underline{\ddot{P}}\left(0\right)+\left(1-w^{+}\right)\overline{\ddot{P}}\left(0\right)$
according to Corollary \ref{cor:combining-probabilities}.}
\end{figure}
\begin{figure}
\includegraphics[scale=0.5]{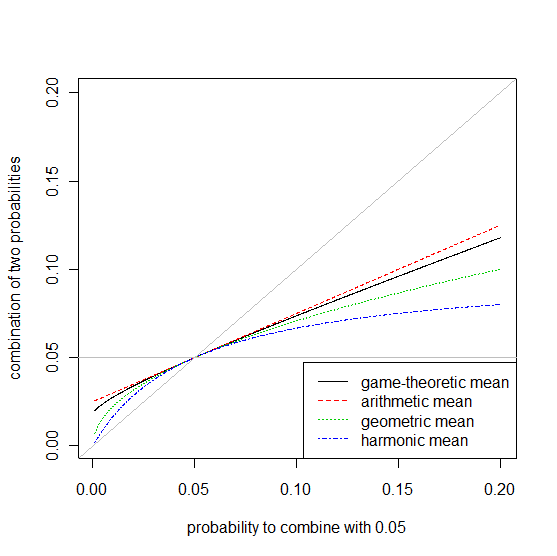}

\caption{\label{fig:means}Three equal-weight averages of probabilities and
the game-theoretic combination of two probabilities based on Corollary
\ref{cor:combining-probabilities}. The two probabilities that are
combined are drawn in solid gray.}
\end{figure}
\begin{figure}
\includegraphics[scale=0.5]{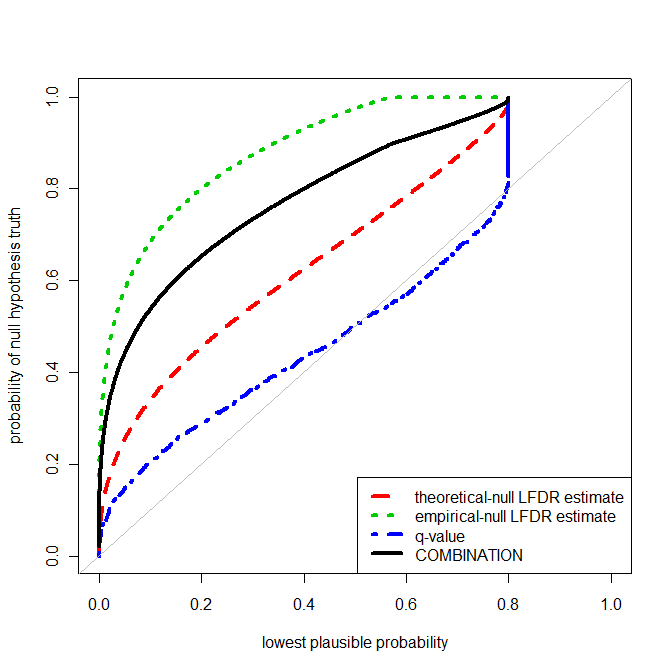}

\caption{\label{fig:Combination-xprn}Combination of estimates of local false
discovery rates, which are empirical Bayes posterior probabilities
that the null hypotheses of equivalent gene expression are true.}
\end{figure}

\end{document}